\documentclass[12pt,onecolumn,a4paper]{IEEEtran}
\IEEEoverridecommandlockouts

\usepackage[utf8x]{inputenc}
\usepackage[english]{babel}
\usepackage{amsmath}
\usepackage{amssymb}
\usepackage{color}
\usepackage{cite} 
\usepackage{amsthm}
\usepackage{comment}
\usepackage{graphicx}
\usepackage{caption}
\usepackage{subcaption}

\newtheorem{ex}{Example}

\newcommand{\Cc}{{\cal C}}
\newcommand{\Mc}{{\cal M}}
\newcommand{\vv}{{\bf v}}
\newcommand{\xv}{{\bf x}}

\begin{document}

\title{Multiuser Detection in Multibeam Satellite Systems: Theoretical Analysis and Practical Schemes}

\author{Giulio Colavolpe,~\IEEEmembership{Senior Member,~IEEE}, Andrea Modenini, Amina Piemontese,~\IEEEmembership{Member,~IEEE}, and Alessandro Ugolini
\thanks{G. Colavolpe and A. Ugolini are with are with Universit\`a di Parma, Dipartimento di Ingegneria dell'Informazione, Viale delle Scienze, 181A, I-43124 Parma, Italy (e-mail: \{giulio, alessandro.ugolini\}@unipr.it).}
\thanks{A. Modenini is  with  ESA-ESTEC, Noordwijk, The Netherlands (e-mail: andrea.modenini@esa.int).}
\thanks{A. Piemontese is with the Department of Signals and Systems, Chalmers University of Technology, Gothenburg, Sweden (e-mail: aminap@chalmers.se).}
\thanks{This work is funded by the European Space Agency, ESA-ESTEC, Noordwijk, The Netherlands. The view expressed herein can in no way be taken to reflect the official opinion of the European Space Agency.}
\thanks{The paper was presented in part at the IEEE Intern. Conf. Commun. (ICC'15), London, UK, June 2015 and at the IEEE Intern. Workshop on Signal Process. Advances in Wireless Commun. (SPAWC'15), Stockholm, Sweden, June 2015.}
% \thanks{Submitted: . Revised:  and~\today.}
}

\maketitle

\begin{abstract}
We consider the rates achievable by a user in a multibeam satellite system for unicast applications, and propose alternatives to the conventional single-user symbol-by-symbol detection applied at user terminals. Single-user detection is known to suffer from strong degradation when the terminal is located near the edge of the coverage area of a beam, and when aggressive frequency reuse is adopted. For this reason, we consider multiuser detection, and take into account the strongest interfering signal. We also analyze two additional transmission strategies requiring modifications at medium access control layer. We describe an information-theoretic framework to compare the different strategies by computing the information rate of the user in the reference beam. Furthermore,  we analyze the performance of coded schemes that could approach the information-theoretic limits. We show that classical codes from the DVB-S2(X) standard are not suitable when multiuser detection is adopted, and we propose two ways to improve the performance, based on the redesign of the code and of the bit mapping.
\end{abstract}
%\begin{keywords}
%Multibeam satellite channels, multiuser detection, information rate analysis.
%\end{keywords}

\section{Introduction}\label{s:intro} 
The recent years have witnessed the explosion of satellite services and applications, and the related growing demand for high data rates. In particular, satellite systems, which are broadcast by nature, can be also used for broadband interactive, and thus unicast, transmissions. The 2nd-generation specification of the digital video broadcasting for satellite (DVB-S2) standard~\cite{DVB-S2-TR}, developed in 2003, and its evolution, approved in 2014 with the name of DVB-S2X~\cite{DVB-S2X}, represent illuminating examples in this sense.

Next-generation satellite systems need new technologies to improve their spectral efficiency, in order to sustain the increasing request of new services. The grand challenge is to satisfy this demand by living with the scarcity of the frequency spectrum. Resource sharing is probably the only option, and can be implemented by adopting a multibeam system architecture, which allows to reuse the available bandwidth in many beams. The interference caused by resource sharing is typically considered undesirable, but a way to dramatically improve the spectral efficiency is to exploit this interference, by using interference management techniques at the receiver. 

In this paper, we consider the benefits of the adoption of multiuser detection at the user terminal in the forward link of a multibeam satellite system. Our reference is a DVB-S2(X) system~\cite{DVB-S2-TR,DVB-S2X}, where an aggressive frequency reuse is applied. Under these conditions, the conventional single-user detector (SUD) suffers from a severe performance degradation when the terminal is located near the edge of the coverage area of a beam, due to the high co-channel interference. On the other hand, the application of a decentralized multiuser detector (MUD) at the terminal, able to cope with the interference, can guarantee the required performance~\cite{AnAnCaCo14,CoAnPe14}. Of course a computational complexity increase must be paid. A parallel investigation on the same topics is reported in~\cite{CaPeAnGi15}.

The literature on multiuser detection is wide, and in the area of satellite communications it essentially focuses on the return link~\cite{PiGrCo13,BeElKa02,CoFePi11, Mo00,LeGr08,ArMo12,ChChKrOt13}, i.e., on the link from the user terminals to the gateway, and includes centralized techniques to be applied at the gateway. Less effort has been devoted to the forward link. Recently, we investigated in~\cite{AnAnCaCo14} the benefits that can be achieved, in terms of spectral efficiency, when high frequency reuse is applied in a DVB-S2 system, and multiuser detection is adopted at the terminal to manage the presence of strong co-channel interference. In~\cite{CoAnPe14}, the authors study the applicability of a low complexity MUD based on soft interference cancellation~\cite{AlGrRe98}. In both papers, the advantage of the proposed detectors is shown in terms of error rate. 

In this paper, we generalize the analysis of~\cite{AnAnCaCo14} by supplying an information-theoretic framework which allows us to evaluate the performance in terms of information rate (IR), without the need of lengthy error rate simulations, and hence strongly simplifying the comparison of various strategies and scenarios. The main results of this investigations are also reported in~\cite{CoMoPiUg15a}. Furthermore, we consider two additional transmission strategies, where the signals intended for the two beams cooperate to serve the two users (one in the first beam and the other in the second one). In one case, the two users in the adjacent beams are served  consecutively in a time division multiplexing fashion, instead of being served simultaneously. This approach is also considered in~\cite{CoMoPiUg15a,CaPeAnGi15}. In the other case, we consider the Alamouti space-time block code~\cite{Al98}, consisting in the two satellites exchanging the transmitted signals in two consecutive transmissions.

Finally, we show that the theoretical limits predicted by the information-theoretic analysis can be approached by practical coded schemes. As expected, the Alamouti precoding based schemes work well with the standard DVB-S2(X) modulation and coding formats (ModCods), designed for an interference-free scenario. On the other hand, we observe that DVB-S2(X) ModCods are not suitable for multiuser detection applications. Therefore, we analyze the convergence behavior of joint multiuser detection/decoding by means of an extrinsic information transfer (EXIT) chart analysis~\cite{te01}. We start by considering DVB-S2(X) ModCods  and quantify the loss with respect to the theoretical limits. Once identified the reasons for this performance loss, we prove that a large gain can be obtained from a redesign of the code and/or of the bit mapper. Part of this investigation is reported in~\cite{CoMoPiUg15b}.

According to the information-theory literature, the multibeam satellite channel is a broadcast channel, with the satellite serving multiple users on the ground. Particularly, we are in the case of a multiple-input multiple-output (MIMO) broadcast channel, since we have multiple antennas at the transmitter (for the different beams).\footnote{This definition of broadcast channel collides with the one commonly adopted in the satellite communication literature.  In fact, in the satellite community, a broadcast transmission refers to the case where one transmitter sends common information to several receivers, in contrast with unicast transmissions considered in this paper.} Nevertheless, we do not use the results concerning the broadcast channel capacity since we are not interested in the ultimate performance limits of the considered scenario. Our work focuses, instead, on the gain that can be achieved by one specific user if it employs a more involved detector, i.e., a MUD, when the receivers of the other users are not necessarily modified. In other words, we want to understand if and when it is worth to use a MUD to decode also the signal which is not intended for the reference user. For this aim, the theory of broadcast channels is not helpful and instead we borrow ideas from the Multiple Access Channel (MAC)~\cite{CoTh06}. Furthermore, it is known that the sum-rate capacity of  the MIMO broadcast channel is achieved by means of dirty-paper coding~\cite{CaSh03}, but nonlinear precoding leads to several problems when going to the practical implementation in satellite systems,  as the channel estimation, the synchronization and the non-linear effects introduced by the satellite amplifier (not necessarily a problem in the case of multicarrier signals). 

In the following, Section~\ref{s:sys_model} presents the system model and describes the three considered strategies and related detection techniques. The information-theoretic analysis is addressed in Section~\ref{s:scenario}, and gives us the necessary means for the computation of the IR for the reference beam. The EXIT chart analysis is described in Section~\ref{s:exit_C}. Section~\ref{s:num_res} presents the results of our study, whereas conclusions are drawn in Section~\ref{s:conclusions}.

\section{System Model and Considered Strategies}\label{s:sys_model}
\begin{figure}[t]
	\centering
	\includegraphics[width=0.6\columnwidth]{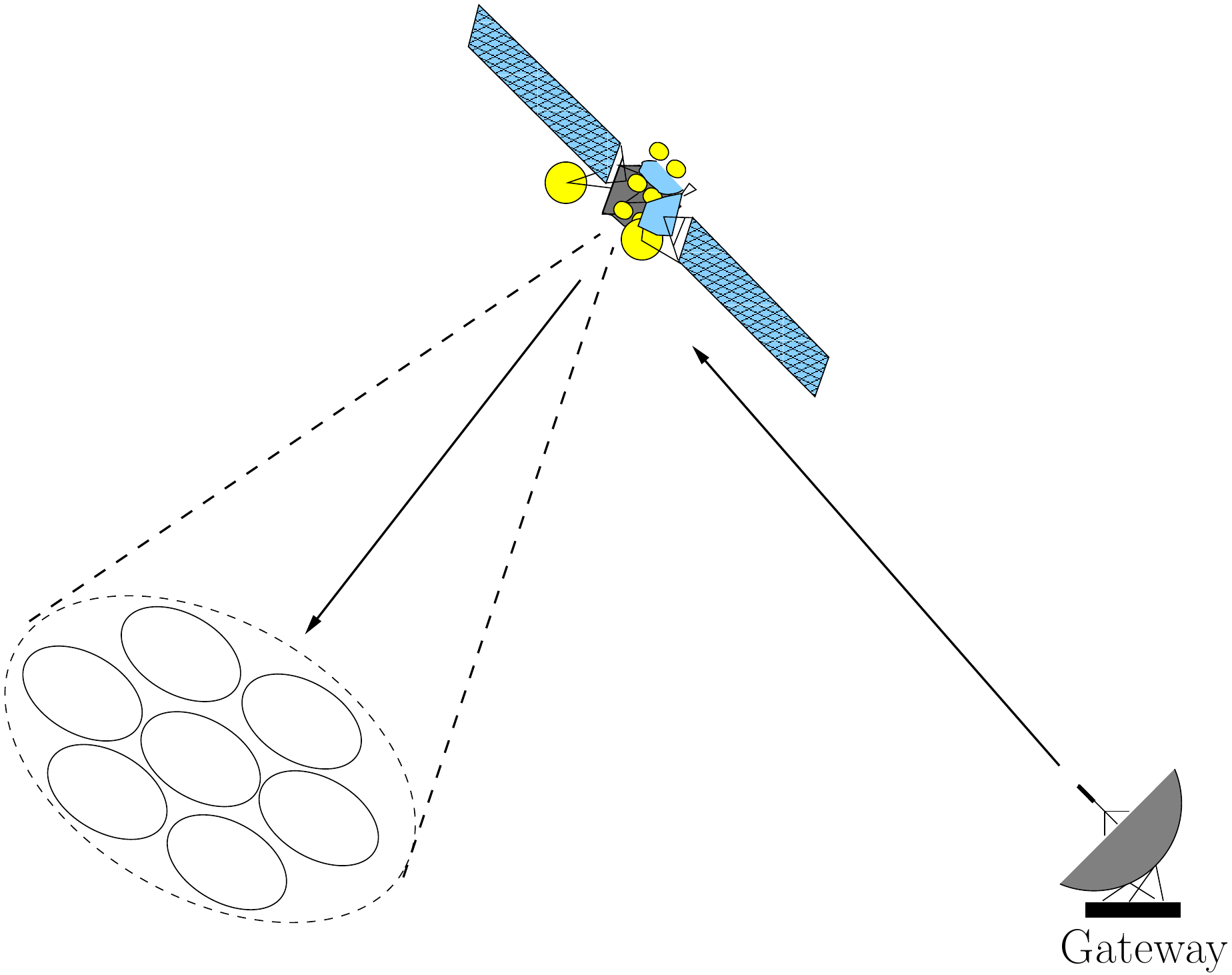}
	\caption{Forward link of a multibeam satellite system. Circles in the satellite service area represent beams.}
	\label{fig:architecture}
\end{figure}
\begin{figure}[t]
	\centering
	\includegraphics[width=0.6\columnwidth]{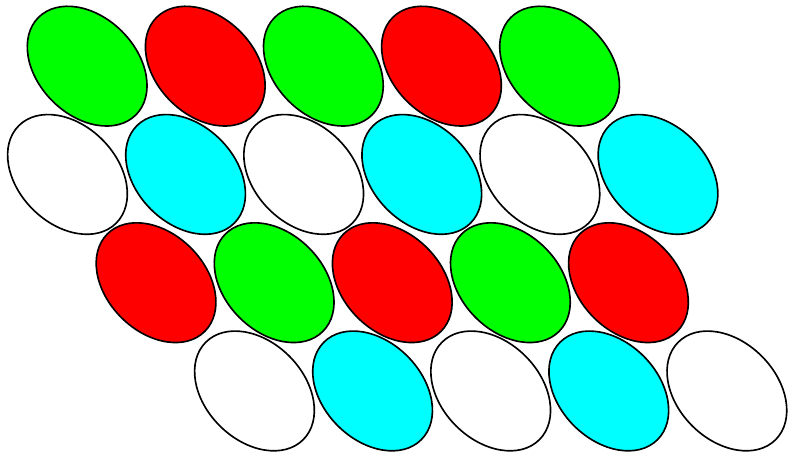}
	\caption{4-color frequency reuse scheme.}
	\label{fig:4c}
\end{figure}
We focus on the forward link (i.e., the link from the gateway to the user terminals through the satellite) of a multibeam satellite communication system for broadband interactive services, as illustrated in Figure~\ref{fig:architecture}. In this scenario, the service area is divided into small beams in order to reuse the frequency spectrum and hence to improve the spectral efficiency.  As an example, a 4-color frequency reuse scheme is shown in Figure~\ref{fig:4c}, where beams with the same color use the same bandwidth. In a 4-color frequency reuse scheme, the interference is very limited and can be neglected at the receiver. Thus, at the receiver, a SUD is employed. A more aggressive frequency reuse can be adopted with the aim of improving the system spectral efficiency. Figure~\ref{fig:2c} depicts the case of a 2-color frequency reuse scheme.  In this latter case, a SUD is still used at the receiver although the interference can be significant.
\begin{figure}[t]
	\centering
	\includegraphics[width=0.6\columnwidth]{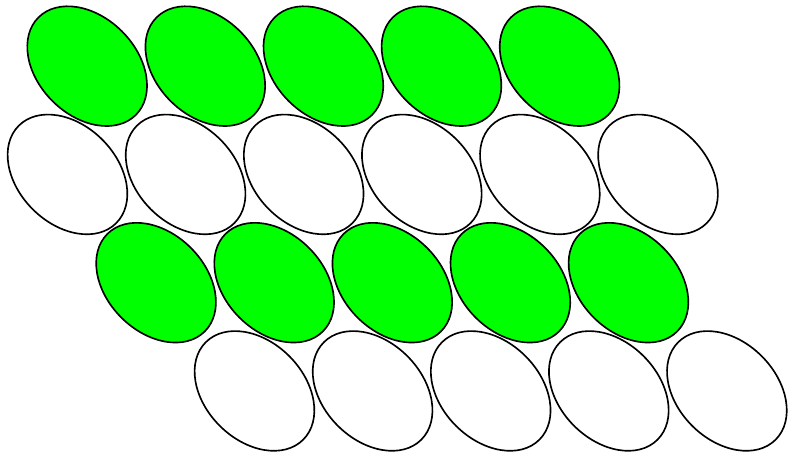}
	\caption{2-color frequency reuse scheme.}
	\label{fig:2c}
\end{figure}	
\begin{figure}
	\begin{center}
		\includegraphics[width=1.0\columnwidth]{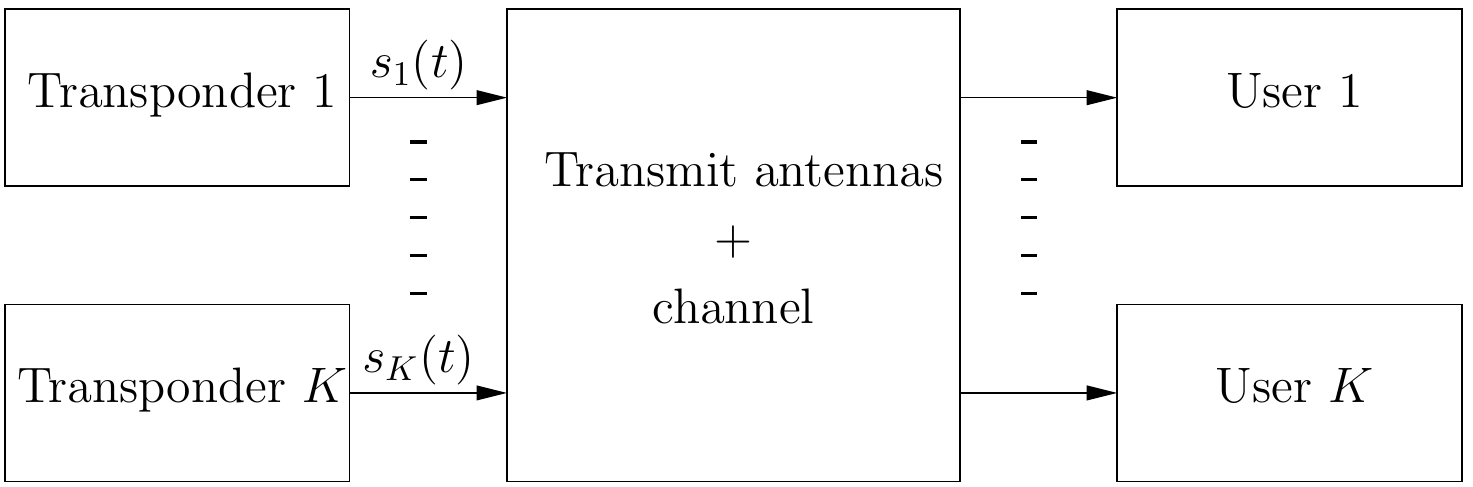}
		\caption{Schematic view of the considered architecture.}\label{fig:system}
	\end{center}
\end{figure}

Assuming an ideal feeder link (i.e., the link between the gateway and the satellite), Figure~\ref{fig:system} depicts a schematic view of the baseband model we are considering. Signals $s_i(t)$, $i=1,\dots,K$, are $K$ signals transmitted by a multibeam satellite in the same frequency band. The satellite is thus composed of $K$ transmitters (i.e., transponders) and serves $K$ users on the ground. The nonlinear effects related to the high power amplifiers which compose the satellite transponders are neglected since a multibeam satellite often works in a multiple carriers per transponder modality, and hence the operational point of its amplifiers is far from saturation~\cite{DVB-S2-TR}. We consider the case where the users experience a high level of co-channel interference, since we assume that they are located close to the edge of the coverage area of a beam and that an aggressive frequency reuse is applied (i.e., a number of colors lower than 4).

The signal received by a generic user can be expressed as
 \begin{equation}
 r(t)=\sum_{i=1}^K\gamma_i s_i(t)+w(t)\,, \label{eq:model}
 \end{equation}
where $\gamma_i$ are proper complex gains, assumed known at the receivers, and $w(t)$ is the thermal noise. Without loss of generality, we assume that ``User 1'' is the reference user and that $r(t)$ is its received signal. We also assume that $\gamma_1=1$, that $|\gamma_i|\ge|\gamma_{i+1}|$, and that the satellite has no way to modify the gains. The satellite could, in principle, change the power for each user, but this is not done in practice for the following reason. In a unicast scenario, if we consider a given frequency bandwidth, different users in a beam are served in time-division-multiplexing mode. Hence, different frames are sent to different users. These users can have different propagation conditions (e.g., some of them can experience rain attenuation) and interference. To take this into account, different ModCods are selected for the different users, in the so-called ACM (Adaptive Coding and Modulation) mode. As a consequence, different frames will use different modulation and coding formats. The gateway could also try to modify the power for each transmitted frame, and thus for each user. However, satellite transponders are equipped with analog automatic gain control circuits, which are very slow. A change in the power, frame by frame, would introduce strange amplitude fluctuations that the system cannot cope with.  Hence, a modification of the power allocation is not an option, at least considering the present transponder architecture.

We will evaluate the performance of the reference user considering the following three strategies, which imply different transmission and detection approaches.
\newline
\textbf{Strategy 1.} Signal $ s_i(t) $ is intended for user $i$, and we are interested in the evaluation of the performance for ``User~1'', whose information is carried by the signal with $i=1$. For this case, we evaluate the IR when ``User 1'' employs different detectors. In particular, we consider the case when ``User 1'' adopts:
\begin{itemize}
\item A SUD. Here, all interfering signals $s_i(t)$, $i=2,\dots,K$ are considered as if they were additional thermal noise. This is what is typically done in present systems.   
\item A MUD for the useful signal and one interferer. In this case, the receiver is designed to detect the useful signal and the most powerful interfering signal (that with $i=2$ in our model), which is assumed to adopt a fixed rate, whereas all the remaining signals are considered as if they were additional thermal noise. Data related to the interfering signal are discarded after detection. This case will be called MUD$\times$2 in the following. The complexity will be clearly larger than that of the benchmark system using a SUD. 

\end{itemize}
The aim of this analysis is to evaluate the performance improvement that can be obtained by simply using a more sophisticated receiver at the user terminal with no modifications of the present standard. In other words, this strategy is perfectly compliant with the DVB-S2(X) standard, since it simply requires the adoption of a different receiver. Our analysis can be easily extended to a MUD designed for more than two users although, given the actual signals' power profile, it has been shown in~\cite{AnAnCaCo14} that the MUD$\times$2 offers the best trade-off between complexity and performance.
\newline
\textbf{Strategy 2.} A different transmission strategy, requiring a modification at medium access control layer with respect to the previous case, is adopted in this case. Hence, in order to adopt this strategy, a modification of the DVB-S2(X) standard is required. Without loss of generality, we will consider detection of signals $s_1(t)$ and $s_2(t)$ and users 1 and 2 only. As in scenario~1, the remaining signals are considered as additional thermal noise. Instead of simultaneously transmitting  signal $s_1(t)$ to ``User~1'' and signal $s_2(t)$ to ``User~2'', as in the previous scenario, we here serve ``User 1'' first by employing both signals $s_1(t)$ and $s_2(t)$ for a fraction $\alpha$ ($0\le\alpha\le 1$) of the total time, and then ``User~2'' by employing both signals $s_1(t)$ and $s_2(t)$ for the remaining fraction $1-\alpha$ of the total time. From a system point of view, in order to maximize the throughput at system level, the best thing to do would be to serve the user with the best channel only. However, this would not be fair, since the user with the worst channel would never be served. Satellite operators are typically interested in serving each of the two users for half of the time or in trying to serve the users taking into account their different data rate needs. Signals $s_1(t)$ and $s_2(t)$ are independent (although carrying information for the same user). The receiver of each user must jointly detect both signals and its complexity is comparable to that of the MUD$\times$2 described for the first scenario.

In strategies 1 and 2, $s_1(t)$ and $s_2(t)$ are properly phase-shifted in order to maximize the IR.\footnote{We assume that the signals are modulated by using exactly the same carrier frequency, i.e., that a common reference oscillator is used for all signals.}
\newline
\textbf{Strategy 3.} As in the first strategy, $s_1(t)$ is for ``User~1'' and $s_2(t)$ for ``User~2''. We use two transponders to implement the Alamouti precoder~\cite{Al98}.  Unlike what happens in~\cite{Al98}, we do not use the Alamouti scheme to achieve a diversity gain, but as a way of orthogonalizing the two signals.  In this scheme, the two transmitters exchange the two information symbols in two successive transmission intervals. At the receiver, two consecutive observed samples are properly processed in order to remove the interference, and then fed to two SUDs. In this way, we can always transmit fully overlapped signals and perform only lower complexity operations at the receiver. To preserve the orthogonality of the two signals, in this approach, the same information has to be transmitted twice over two consecutive intervals. The IR is thus halved for this reason.

\section{Information-theoretic Analysis}\label{s:scenario}
In this section, we describe how to compute the IR related to ``User 1'' assuming the previously described strategies. This analysis gives us the ultimate performance limits of the considered satellite system, which will be used as a benchmark for the performance of practical coded schemes.

We start considering \textbf{strategy~1}, and describe how to compute the IR of ``User 1'' assuming the MUD$\times$2 receiver. The same technique can be used to compute the IR related to ``User 2'' and straightforwardly extends to the case of MUD for more than two users. The channel model assumed by the receiver is
	\begin{equation}
		y = x_1+ \gamma_2 x_2 + w \, ,\label{eq:mac}
	\end{equation}
where $x_i$ is the $M^{(i)}$-ary complex-valued symbol sent over the $i$-th beam and $w$ collects the thermal noise, with power $N$, and the remaining interferers that the receiver is not able to cope with. Symbols $x_1$ and $x_2$ are mutually independent and distributed according to their probability mass function (pmf) $P(x_i)$. They are also properly normalized such that $\mathrm{E}\{|x_i|^2\}=P$, where $P$ is the transmitted power per user. The parameter $\gamma_2$ is complex-valued and models the power imbalance and the phase shift between the two signals. The random variable $w$ is assumed complex and Gaussian. We point out that this is an approximation exploited only by the receiver, while in the actual channel the interference is clearly generated as in~\eqref{eq:model}. The MUD$\times$2 receiver has a computational complexity which is proportional to the product $M^{(1)}  M^{(2)}$~\cite{Ve98}.

We are interested here in the computation of the maximum achievable rate $R_1$ for ``User 1'' when ``User~2'' adopts a fixed rate $R_2$, and the MUD$\times$2 is employed. Rates are defined as $R_i=r^{(i)}\log_2\left(M^{(i)}\right)$, where $r^{(i)}$ is the rate of the adopted binary code. The rates of the other $K-2$ interferers do not affect our results since, at the receiver, they are treated just as noise. This problem is quite different with respect to the case of the MAC discussed in~\cite{CoTh06} where both rates $(R_1,R_2)$ are jointly selected. In this case, in fact, the information coming from the second transmitter is not intended for ``User 1''. Hence, the rate $R_2$ is fixed and data of ``User~2'' is discarded after detection.

The IR for ``User~1'' in the considered scenario is given by Theorem~\ref{t:ThIR}, whose proof is based on the following two lemmas. An alternative proof can be found in~\cite{NaBaLeKa14}. The first one defines the maximum rate $I_A$ achievable by ``User 1'' when ``User 2'' can be perfectly decoded.
	\newtheorem{Lemma}{Lemma}
	\begin{Lemma}\label{l:lemma1}
		For a fixed rate $R_2$, the rate
		\begin{equation}\nonumber
			I_{A}= \begin{cases}
				I(x_1;y|x_2) & \mathrm{if\quad} R_2 < I(x_2;y) \\
				 I(x_1,x_2;y)-R_2 & \mathrm{if\quad} I(x_2;y) \leq R_2<I(x_2;y|x_1) \\
			         0                                          & \mathrm{if\quad} R_2\geq I(x_2;y|x_1)
			       \end{cases}
		\end{equation}
is achievable by ``User 1'' and is not a continuous function of $P/N$. Namely, a cut-off signal-to-noise ratio $\mathrm{SNR}_c$  exists such that $I_A=0$ for $P/N \leq \mathrm{SNR}_c$ and $I_A>0$ for $P/N > \mathrm{SNR}_c$ with a discontinuity.
	\end{Lemma}
	\begin{proof}
		In \cite{CoTh06}, it is shown that the achievable region for the MAC is given by the region of points ($R_1,R_2$) such that
		\begin{eqnarray}
			R_1 & <& I(x_1;y|x_2)=I_1 \label{eq:c1} \\
			R_2 & <& I(x_2;y|x_1)=I_2 \\
			R_1+R_2 & <& I(x_1,x_2;y)=I_{\rm J}\label{eq:c3}  \,.
		\end{eqnarray}
An example of such a region is shown in Figure~\ref{fig:region}.
\begin{figure}
	\begin{center}
		\includegraphics[width=0.7\columnwidth]{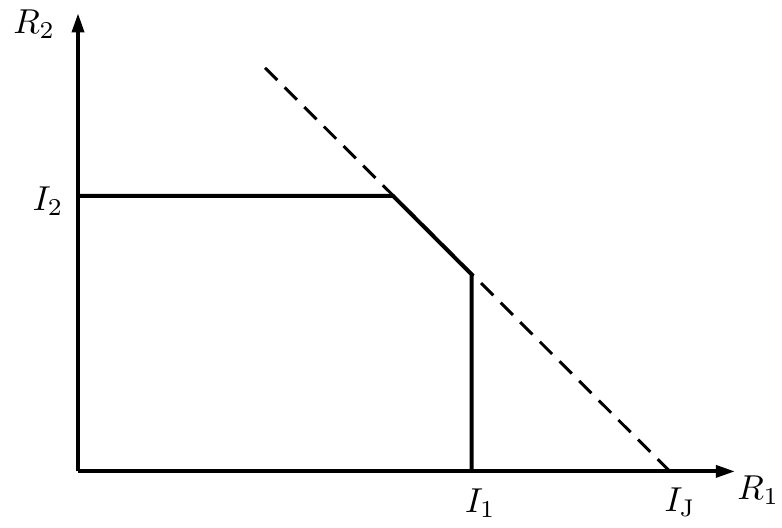}
		\caption{Example of MAC capacity region.}\label{fig:region}
	\end{center}
\end{figure}
If $R_2$ is constrained to a given value, we derive from (\ref{eq:c1}) and (\ref{eq:c3}) that 
$$R_1< \min\{ I(x_1;y|x_2),I(x_1,x_2;y)-R_2\}$$ 
when $R_2<I(x_2;y|x_1)$. The first term is lower when 
$$R_2 < I(x_1,x_2;y) -I(x_1;y|x_2) = I(x_2;y)\,.$$ 
Thus, $I_{A}$ is an achievable rate for ``User 1''.
		
We now prove that $I_A$ has a cut-off rate. Since, $I(x_2;y|x_1)$ is a non-decreasing function of $P/N$~\cite{GuShVe05}, there exists $\mathrm{SNR}_c$ such that $I(x_2;y|x_1)=R_2$, and hence $$I_A(\mathrm{SNR}_c)=0.$$ On the other hand, for a small $\varepsilon>0$,  it holds $R_2=I(x_2;y|x_1)-\delta$ where $\delta>0$. It follows that 
$I(x_1;y|x_2)>I(x_1,x_2;y)-R_2$. Thus 
$$I_A(\mathrm{SNR}_c+\varepsilon)=I(x_1,x_2;y)-R_2> I(x_1;y)>0$$ 
for $\varepsilon\rightarrow0^+$.
\end{proof}
\textit{Discussion}: The proof of the lemma can be done graphically by considering the intersection of the achievable rate region with a horizontal line at height $R_2$. 

When $R_2>I(x_2;y|x_1)$ clearly the rate of ``User 2'' cannot be achieved. However, we also have to account for this case, and therefore we consider also the achievable rate $I(x_1;y)$, which is the relevant rate when ``User 2'' is just considered as interference. In this case, the receiver exploits the statistical knowledge of the signal $s_2(t)$ but does not attempt to recover the relevant information. Particularly, the receiver does not include the decoder for ``User~2''.
\begin{Lemma}\label{l:lemma2}
		The rate $I_S(P/N)=I(x_1;y)$ as a function of $P/N$ is always greater than {\rm 0} and satisfies
		\begin{eqnarray*}
			I_S(\mathrm{SNR}_c) & = & \lim_{\varepsilon \rightarrow0^+} I_A(\mathrm{SNR}_c+\varepsilon) \\
			I_S(\mathrm{SNR}_c+\delta) & < & I_A(\mathrm{SNR}_c+\delta)
		\end{eqnarray*}
		for any $\delta>0$.

	\end{Lemma}
	\begin{proof}
		The proof is straightforward. It can be done by observing that $I(x_1;y)\leq I(x_1;y|x_2)$ and that $I(x_1;y)\leq I(x_1,x_2;y)$.
	\end{proof}

The computation of the IRs $I(x_1;y|x_2)$, $I(x_2;y|x_1)$, $I(x_1,x_2;y)$, $I(x_1;y)$ in the presence of interferers with $i>2$ not accounted for at the receiver can be performed by using the achievable lower bound based on mismatched detection~\cite{MeKaLaSh94}. Having defined $I_A$ and $I_S$ as the maximum rates achievable by ``User~1'' when the signal for the other user can be perfectly decoded, or not, we can now compute the IR for ``User~1'' by means of the following theorem.
	
	\newtheorem{Theorem}{Theorem}
	\begin{Theorem}\label{t:ThIR}
		The achievable information rate for a single user on the two users multiple access channel, for a fixed rate $R_2$, is given by
		\begin{equation}\label{eq:thm1}
			R_1 < \max\{I_S,I_A\}\, ,
		\end{equation}				
		and is a continuous function of $P/N$.
	\end{Theorem}
	\begin{proof}
	 The proof is based on Lemma~\ref{l:lemma1} and \ref{l:lemma2}. In fact, $I_A$ and $I_S$ are the maximum rates achievable by ``User 1'' when the signal for ``User 2'' can be perfectly decoded, or not. An alternative graphical proof can be derived from Figure~\ref{fig:region_thm}, which plots the rate achievable by ``User~1'' as a function of $R_2$, for a generic fixed value of $P/N$. We clearly see that inequality~\eqref{eq:thm1} is satisfied.
	\end{proof}
\begin{figure}
	\begin{center}
		\includegraphics[width=0.7\columnwidth]{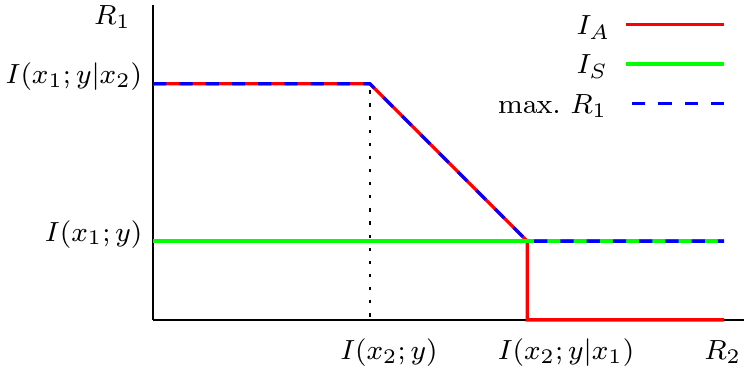}
		\caption{Graphical proof of Theorem~\ref{t:ThIR}.}\label{fig:region_thm}
	\end{center}
\end{figure}

\begin{ex} For Gaussian symbols and $K=2$, we obtain that 
\begin{equation}\nonumber
R_1 < \begin{cases}
		\mathcal{G}\left(\frac{P}{N} \right) &  \mathrm{if\quad}R_2<\mathcal{G}\left(\frac{P|\gamma_2|^2}{N+P} \right) \\
	 \mathcal{G}\left(\frac{P(1+|\gamma_2|^2)}{N} \right)\!-\!R_2 	& \mathrm{if\quad} \mathcal{G}\left(\frac{P|\gamma_2|^2}{N+P} \right) \leq R_2<\mathcal{G}\left(\frac{P|\gamma_2|^2}{N} \right)\\
	\mathcal{G}\left(\frac{P}{N+P|\gamma_2|^2} \right) & \mathrm{if\quad}R_2\geq \mathcal{G}\left(\frac{P|\gamma_2|^2}{N} \right)\,,
      \end{cases}
\end{equation}
where $\mathcal{G}(x)=\log_2(1+x)$. All curves are shown in Figure \ref{fig:GaussCurve}, for the case of $|\gamma_2|=0.79$, $R_2=1/2$, and the overall bound is given by the red curve. We can see from the figure that this bound is clearly continuous.
\end{ex}

\begin{figure}
	\includegraphics[width=1.0\columnwidth]{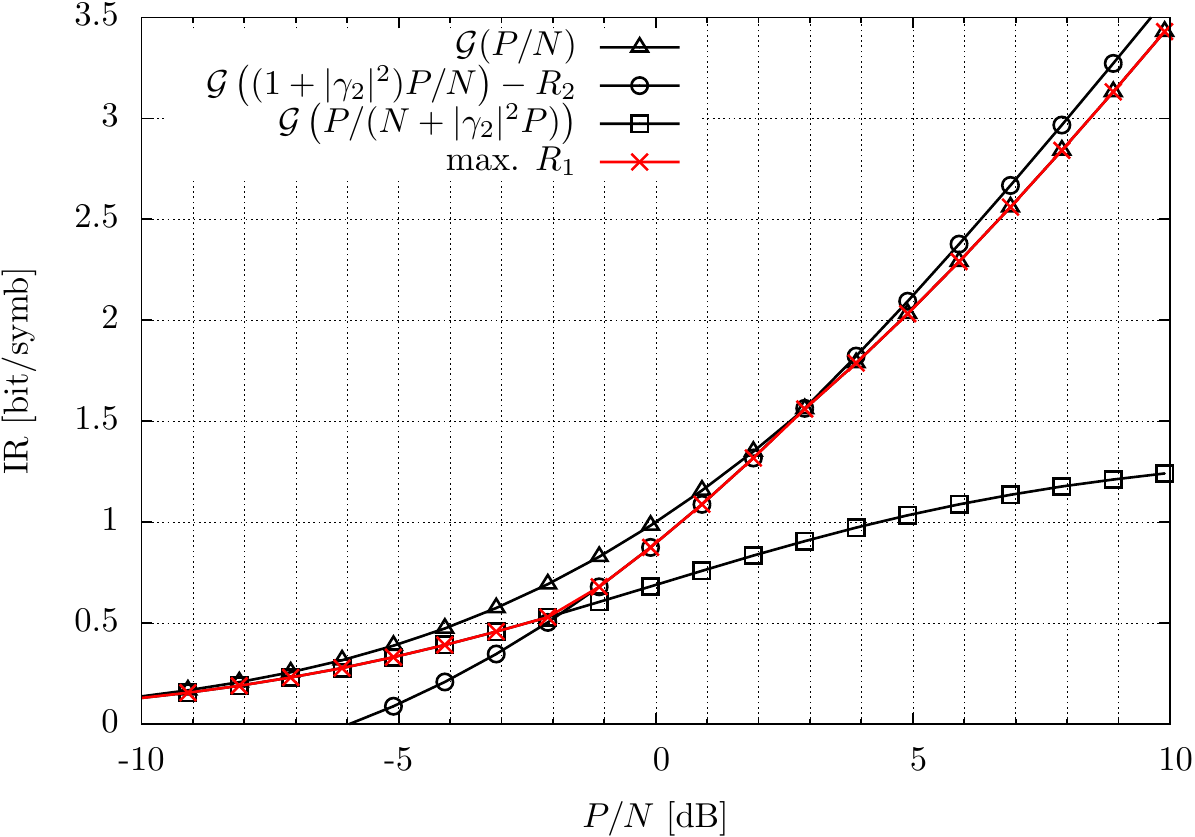}
	\caption{Maximum rate achievable by ``User 1'', for $K=2$, Gaussian symbols, and $R_2=1/2$.}\label{fig:GaussCurve}
\end{figure}

When a SUD is employed at the terminal, the theoretic analysis can be based on the following discrete-time model 
	\begin{equation}
		y = x_1+ w \, , \label{eq:channel_SUD}
	\end{equation}
where $w$ includes the thermal noise and the interferers that the receiver ignores. Note that we again use mismatched detection~\cite{MeKaLaSh94} here, i.e., in the Montecarlo average to compute $I(x_1;y)$ we use the received samples $y$ coming from the real channel whereas the detector is assumed to be designed for the auxiliary channel model~(\ref{eq:channel_SUD}).
As known, the complexity of the SUD is much lower than that of the multiuser receiver, and is proportional to $M^{(1)}$. The computation of the IR $I(x_1;y)$ allows us to select the maximum rate for ``User~1'' when the co-channel interference is not accounted for.

We now consider \textbf{strategy~2} and, without loss of generality, we consider the fraction $\alpha$ of time when both signals $s_1(t)$ and $s_2(t)$ are used to send information to ``User 1''. The receiver is based on the channel model~(\ref{eq:mac}), but now the rate of signal $s_2(t)$ is not fixed. Since $s_1(t)$ and $s_2(t)$ are independent, we are exactly in the case of the MAC and, by properly selecting the rate of the two signals, any point of the capacity region can be achieved~\cite{CoTh06}. Clearly, since $s_1(t)$ and $s_2(t)$ are now both intended for the same user, we are interested in selecting the two rates in such a way that the sum-rate $I(x_1,x_2;y)$ is maximized. 

In \textbf{strategy~3}, due to the adoption of the Alamouti precoding, two consecutive samples at the terminal of ``User 1'' are~\cite{Al98}
\begin{eqnarray*}
	y_{{\rm A},1} & = & x_1 + \gamma_2 x_2 + w_{{\rm A},1} \\
	y_{{\rm A},2} & = & -x_2^* + \gamma_2 x_1^* + w_{{\rm A},2} \, ,
\end{eqnarray*}
where $w_{{\rm A},1},w_{{\rm A},2}$ include independent Gaussian noise samples and the remaining interferers. After the receiver processing~\cite{Al98}, the observable for detection is
\begin{equation}
	\tilde{y}_{{\rm A},i}= \sqrt{1+|\gamma_2|^2}x_i + \tilde{w}_{{\rm A},i} \quad i=1,2 \,, \label{eq:al_eq}
\end{equation}	 	
and is still a sufficient statistic for detection. The noise samples $\tilde{w}_{{\rm A},i}$ are statistically equivalent to $w_{{\rm A},i}$. The information carried by $\tilde{y}_{{\rm A},2}$ is discarded and the IR for ``User~1'' is that of an interference free channel with SNR $(1+|\gamma_2|^2)P/N$, divided by 2 for the reason already explained.

\section{EXIT Chart Analysis}\label{s:exit_C}
In this section, we analyze the convergence behaviour of the considered strategies based on multiuser detection by means of an EXIT chart analysis~\cite{te01}. The aim is to evaluate the effectiveness of iterative decoding/detection schemes in \textbf{strategy 1} and \textbf{strategy 2}, and to design novel practical systems with performance close to the theoretical limits. 

In the following description, we assume the presence of only two independent signals, those processed by the receiver of ``User~1'', but the results in Section~\ref{s:num_res} are generated according to the general model~(\ref{eq:model}). 

Each transmitted signal is obtained through a concatenation of a code with a modulator through a bit interleaver. The information data of signal $i$ is encoded by encoder $\Cc_i$ of rate $r^{(i)}$ into codeword $\vv_i$, which is interleaved and mapped through a modulator $\Mc_i$ onto a sequence of $M^{(i)}$-ary symbols $\xv_i$. Here the channel model is the vectorial extension of the model~(\ref{eq:mac}) which allows us to consider sequences of symbols.  The iterative decoding/detection scheme consists of a multiuser detection module $\Cc^{-1}_\mathrm{MU}$, and 2 \textit{a posteriori} probability decoders $\Cc^{-1}_1$ and $\Cc^{-1}_{2}$ matched to encoders $\Cc_1$ and $\Cc_2$ of the two transponders. The described system is reported in Figure~\ref{fig:iterative_scheme}.

The soft-input soft-output (SISO) MUD exchanges soft information with the two decoders $\Cc^{-1}_1$ and $\Cc^{-1}_2$, in an iterative fashion. More generally, the detector and the decoders can also be composed by SISO blocks. In this work, we focus on low-density parity-check (LDPC) codes, whose decoder is composed of sets of variable and check nodes (the variable-node decoder (VND) and check-node decoder (CND)). Iterative decoding is performed by passing messages between variable and check nodes. 

The global iterative detection/decoding process can then be tracked using a multi-dimensional EXIT chart~\cite{te01}. Alternatively, the EXIT functions of the constituent decoders and of the MUD can be properly  combined and projected into a two-dimensional chart~\cite{Bra05IT}. Similar to a system composed by only two SISO blocks, the convergence threshold of our system can be visualized as a tunnel between the two curves in the projected EXIT chart.

The system in Figure~\ref{fig:iterative_scheme} can represent both \textbf{strategy~1} and \textbf{strategy~2}. We recall that in the first scenario the information to recover is conveyed by signal~1 only, while the rate of the other signal is fixed. Our design will be thus aimed at finding a good code $\Cc_1$, while the code for the other signal cannot be changed and will be chosen among those foreseen by the DVB-S2(X) standard. For \textbf{strategy~2}, the scheme in Figure~\ref{fig:iterative_scheme} is representative of the fraction of time in which both signals are carrying information for ``User~1''. In this case, we assume to have the freedom to choose the code of the two signals and also to apply a joint bit mapping, as we will see in Section~\ref{sec:joint mapping}. 

\begin{comment}
Moreover, if we assume that $|\gamma_2|=1$, the convergence analysis strongly simplifies since it is reasonable to assume that the two transponders employ the same code. In this special case, we have that
$$I_{E^{\Cc_1}_{\vv_1}}=I_{E^{\Cc_2}_{\vv_2}}$$
and the EXIT chart can be directly visualized in two dimensions, by plotting the MI curve of the MUD and of the decoder.
\end{comment}
\begin{figure}
	\begin{center}
		\includegraphics[width=1.0\columnwidth]{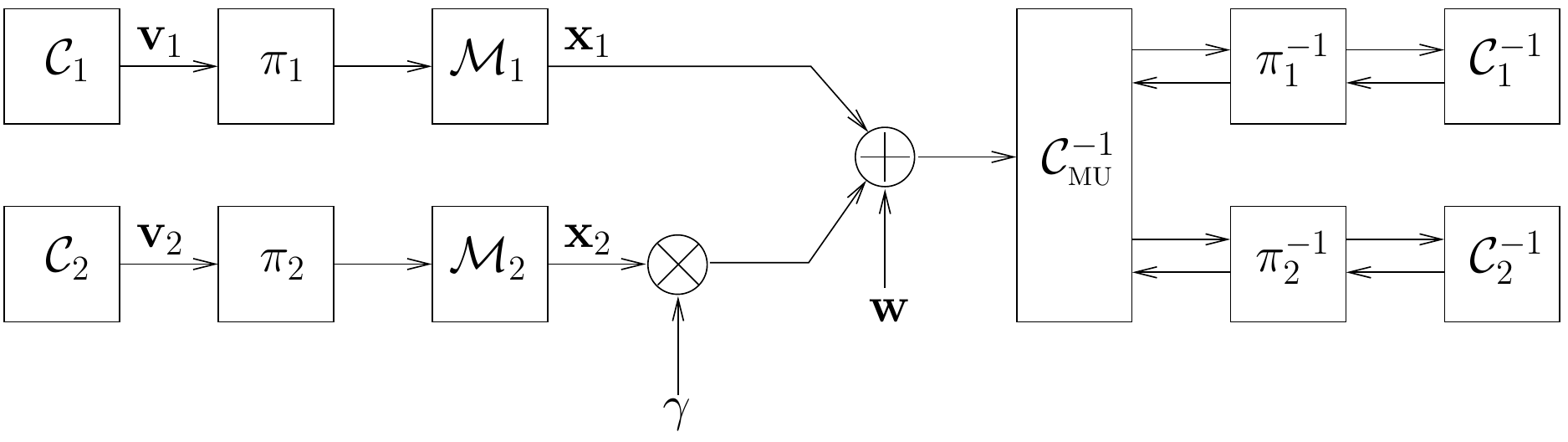}
		\vspace{-0.5cm}
		\caption{Block diagram of the considered system.}\label{fig:iterative_scheme}
		\vspace{-0.5cm}
	\end{center}
\end{figure}

\section{Numerical Results}\label{s:num_res}
In this section, we will first compare the three described strategies in terms of IR under different conditions. We will then try to approach the information-theoretic results with practical modulation and coding formats. 

\subsection{Information-Theoretic Analysis}
We assume as reference system the DVB-S2 standard~\cite{DVB-S2-TR}. We choose a 2-color frequency reuse scheme  to generate a high co-channel interference. We assume that ``User 1'' is located close the  edge of the coverage area of its beam. To identify the interference, we define the signal-to-interference power ratio as
$$
\lambda_i=|\gamma_1|^2/|\gamma_i|^2\, ,
$$
and consider three realistic cases which have a different power profile, and are listed in Table~\ref{t:cases}. These distributions correspond to 3 different positions for ``User 1'' and are typical of the forward link of a multibeam broadband satellite system with 2-color frequency reuse.

\begin{table}  \caption{Interference profiles corresponding to a 2-color frequency reuse.}   \label{t:cases}
% \vspace{-0.3cm}
\begin{center}
\begin{tabular}{|c|c|c|c|c|c|}
  \hline 
  Case & $\lambda_2$ & $\lambda_3$ & $\lambda_4$ & $\lambda_5$ & $\lambda_6$  \\ 
  \hline \hline
   1 &  0 dB & 25 dB & 25 dB & 27 dB & 30 dB \\ 
  \hline 
   2  & 2 dB & 26 dB & 26 dB & 27 dB & 30 dB\\ 
  \hline
   3 & 4 dB & 27 dB & 26 dB & 27 dB & 30 dB \\ 
  \hline 
  \end{tabular}
\end{center}
\end{table}

For the first two strategies, we assume that ``User~1'' adopts a QPSK modulation, therefore the signal with $i=1$ in \textbf{strategy~1}, and signals 1 and 2 in \textbf{strategy~2} use a QPSK. This is reasonable since we are considering the presence of a strong interfering signal and thus a modulation with a low cardinality will be selected. ``User~1'', in \textbf{strategy 3}, adopts a 16APSK modulation so that we have the same receiver complexity as in \textbf{strategy~2}. In the case of \textbf{strategy~1}, the performance is heavily affected by the rate of ``User~2'': in order to fix the rate of signal~2, we consider the ModCod distribution shown in Figure~\ref{fig:histogram}. The other signals adopt the following modulation formats in all strategies: 8PSK for signals with $i=3,4$ and 6, and 16APSK for the signal with $i=5$ (although only their power really matters). In the case of \textbf{strategy~2}, $\alpha=0.5$ is assumed.

\begin{figure}
	\begin{center}
		\includegraphics[width=1.0\columnwidth]{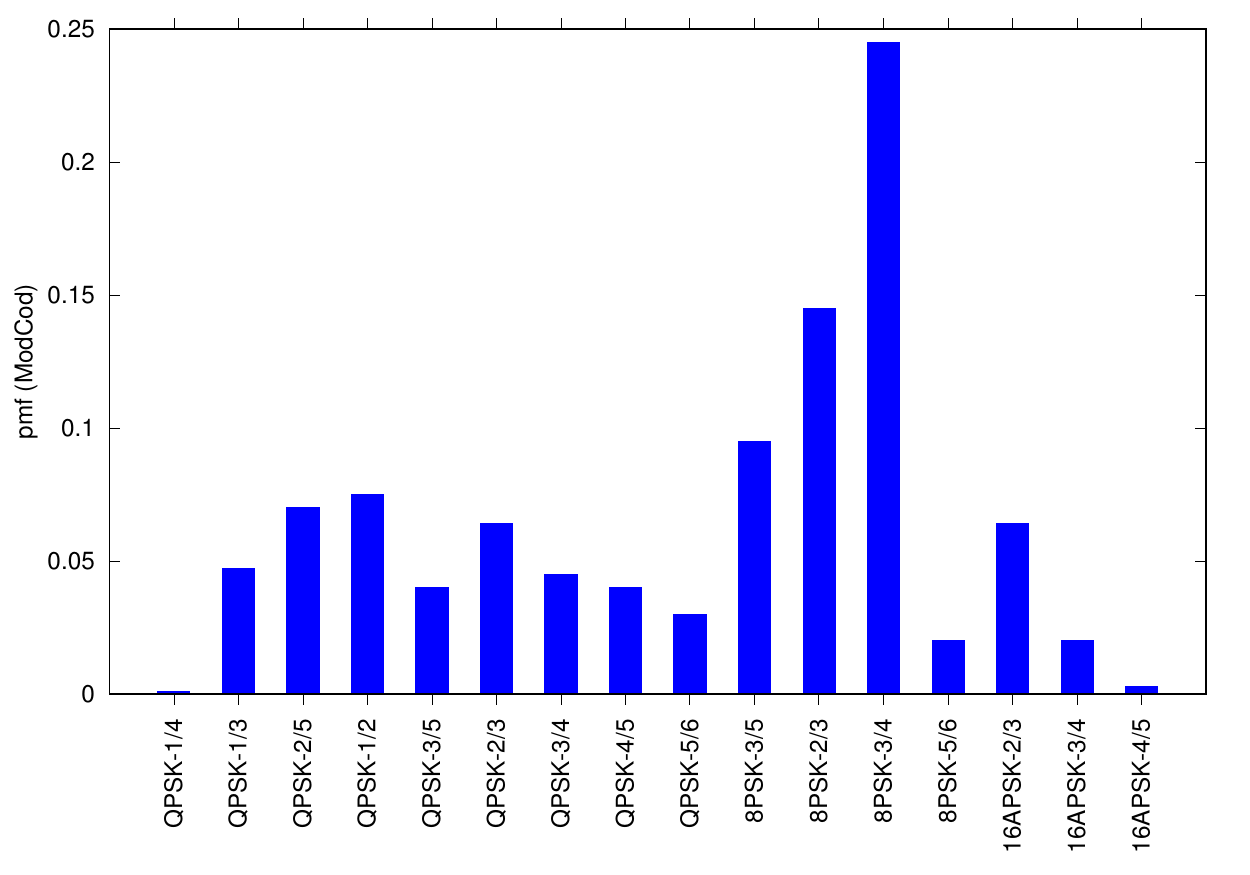}
		\vspace{-0.8cm}
		\caption{Typical ModCod distribution.}\label{fig:histogram}
	\end{center}
\end{figure}
Figures~\ref{fig:IR1}--\ref{fig:IR3} show the IR, measured in bit per symbol, of ``User 1'' as a function of $P/N$ for the three considered interference profiles listed in Table~\ref{t:cases}. In the case of \textbf{strategy~1}, we evaluate both the IR achievable by a SUD and that achievable by the MUD$\times$2 algorithm, and the reported curves are obtained by computing the IRs when ``User~2'' adopts the ModCods in Figure~\ref{fig:histogram} and then averaging according to their distribution.
\begin{figure}
	\begin{center}
		\includegraphics[width=1.0\columnwidth]{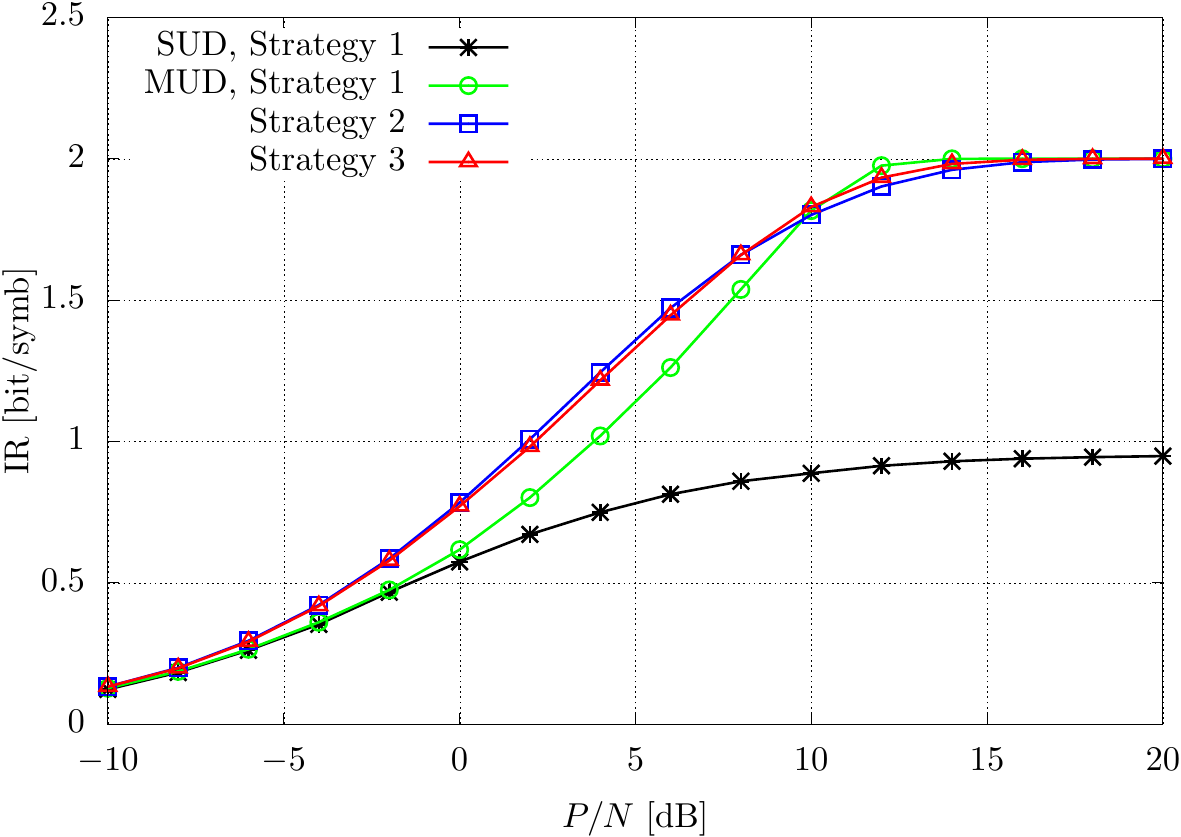}
		\caption{Information rate of ``User 1'' for case 1.}\label{fig:IR1}
	\end{center}
\end{figure}

\begin{figure}
	\begin{center}
		\includegraphics[width=1.0\columnwidth]{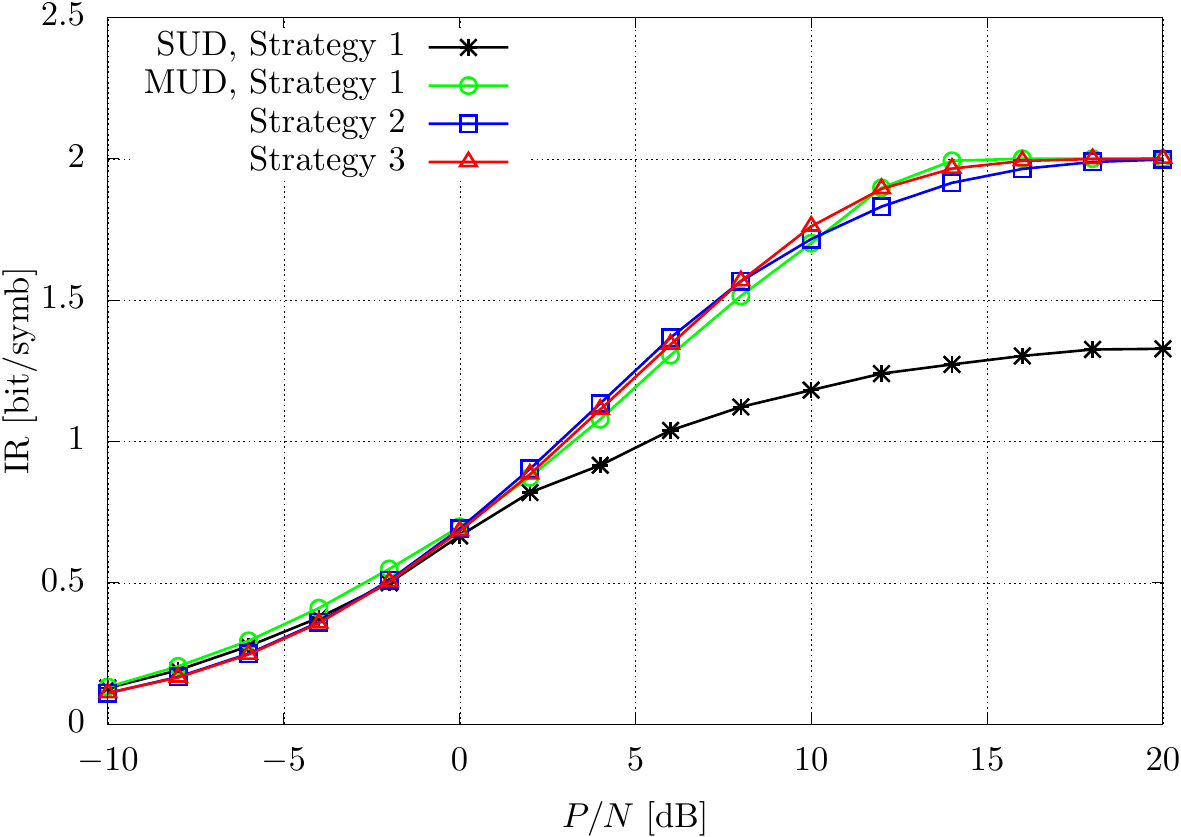}
		\caption{Information rate of ``User 1'' for case 2.}\label{fig:IR2}
	\end{center}
\end{figure}

\begin{figure}
	\begin{center}
		\includegraphics[width=1.0\columnwidth]{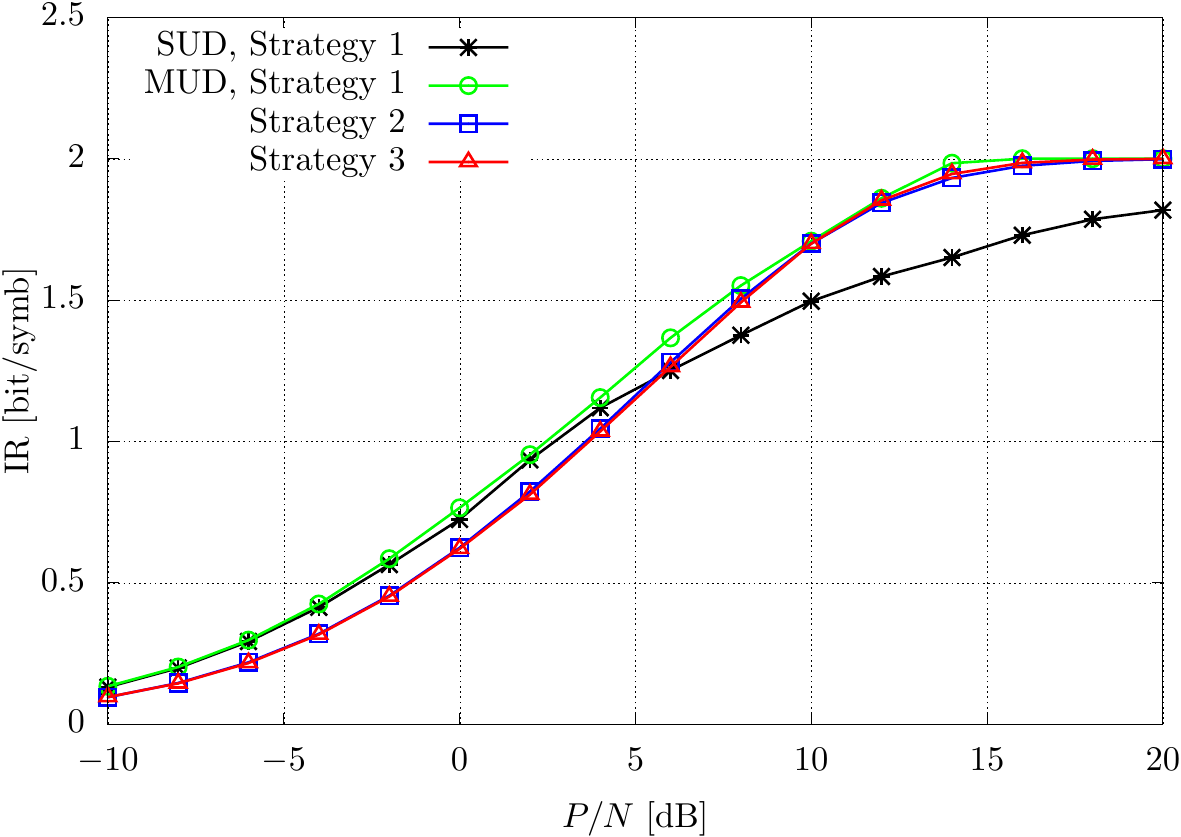}
		\caption{Information rate of ``User 1'' for case 3.}\label{fig:IR3}
	\end{center}
\end{figure}

\begin{figure}
	\begin{center}
		\includegraphics[width=1.0\columnwidth]{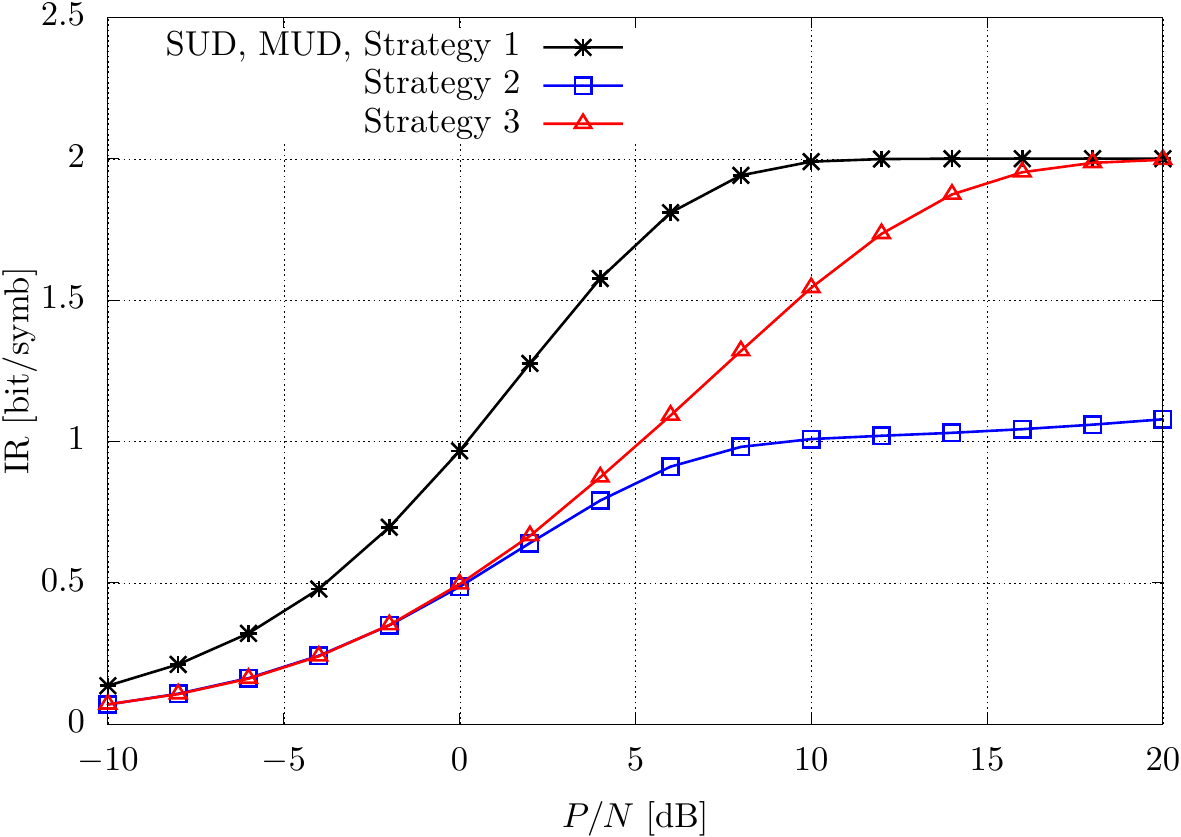}
		\caption{Information rate of ``User 1'' when it is located at the center of the beam (case 4).}\label{fig:IR4}
	\end{center}
\end{figure}

Our results show that we cannot identify the strategy which universally achieves the best performance. In particular, the figures show that ``User~1'' has the best IR in \textbf{strategy~2} and \textbf{strategy~3} in the first case,  where the interference of the second signal is very strong, while in the third case it has higher IR in \textbf{strategy~1} for low-to-medium SNR values. In the second case the three strategies offer similar performance when the MUD is applied in \textbf{strategy~1}. As expected, in \textbf{strategy~1} the adoption of the MUD gives a better result than the SUD, and this is at the price of an increased complexity. In case~3, the SUD gives very good IRs and hence it is the best choice to compromise between complexity and performance for a large SNR range. 

It is worth noting that, while all the proposed strategies are very effective when ``User~1'' is close to the edge of the beam, this is not always true if the user is located at the center of the beam. In Figure~\ref{fig:IR4}, we report the IR in a case in which the power of the signal $s_2(t)$ is very low. The related signals' power profile is given in Table~\ref{t:case4}. In this case, the best strategy is the classical one, and the IR in the case of \textbf{strategy~2} is highly degraded since half of the data for ``User~1'' cannot be recovered, due to the very low value of $|\gamma_2|$. This fact calls for a performance evaluation at system level.
\begin{table}  \caption{Interference profile when ``User~1'' is at the center of the beam.}   \label{t:case4}
% \vspace{-0.3cm}
\begin{center}
\begin{tabular}{|c|c|c|c|c|c|}
  \hline 
  Case & $\lambda_2$ & $\lambda_3$ & $\lambda_4$ & $\lambda_5$ & $\lambda_6$  \\ 
  \hline \hline
   4 & 27 dB & 27 dB & 26 dB & 27 dB & 30 dB \\ 
  \hline 
  \end{tabular}
\end{center}
\end{table}

\subsection{Code Design}
We now consider practical ModCods for multiuser detection and for the Alamouti precoder, and we focus on the gap between practical and theoretical performance. In the following, we do not consider the SUD of \textbf{strategy~1}. As shown by the information-theoretic analysis, it is not easy to compare the three strategies, since the best strategy depends on the power profile of the interfering signals, the rates of the signals, and the SNR. Figure~\ref{fig:IR} shows the IR in case~1. In \textbf{strategy~1}, the IR curve is no more the average IR with respect to the distribution in Figure~\ref{fig:histogram}, but the signal $s_2(t)$ is assumed to adopt an 8PSK. We first consider ModCods based on the LDPC codes of rate 1/2 and 3/4\footnote{The adoption of these two code rates for ``User~1'' corresponds to IR 1 and 1.5 bit/symbol, respectively.} with length 64800 bits of the DVB-S2 standard, with the related interleavers. In \textbf{strategy~1}, we use the rate 3/4 LDPC code for signal $s_2(t)$, in order to 
simulate the 
most probable ModCod according to the distribution in Figure~\ref{fig:histogram}. In the first two strategies, we consider iterative detection and decoding and allow a maximum of 50 global iterations. The BER results have been computed by means of Monte Carlo simulations and are reported in the IR plane in Figure~\ref{fig:IR} using, as reference, a BER of $10^{-4}$. 

These results show that schemes based on the Alamouti precoding and the codes of the standard have good performance, being the loss with respect to the corresponding IR curve around 1~dB. This is because interference is perfectly removed at the receiver. On the contrary, the loss of practical ModCods with respect to the IR limits is high for both \textbf{strategies 1 and 2}, being about 2 and 4~dB at $\textrm{IR}=1$ and 1.5 bit/symbol, respectively.
This is due to the fact that DVB-S2(X) codes have been optimized for an interference-free scenario.

In the following sections, we try to reduce this loss by redesigning the code of ``User~1''. Furthermore, we propose a bit mapping which is jointly implemented for signals 1 and 2 in \textbf{strategy~2}, where we have greater design freedom since both signals are for ``User~1''. Our design approach is based on EXIT charts: this tool is able to point out the limits of the DVB-S2 based ModCods and provide very useful insights on the code and mapper design. 

\begin{figure}
	\begin{center}
		\includegraphics[width=1.0\columnwidth]{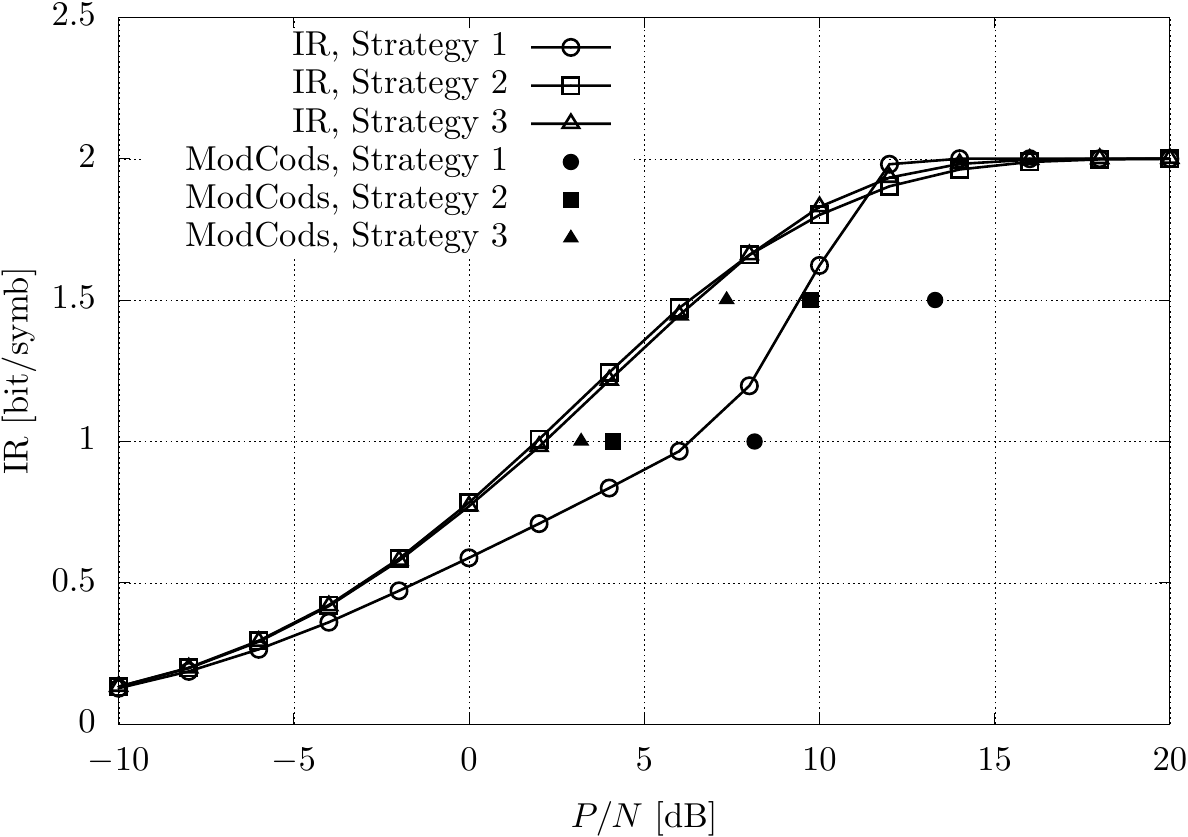}
		\vspace{-0.5cm}
		\caption{IR and ModCods of ``User 1'' in case~1. ModCods are based on DVB-S2 LDPC codes.}\label{fig:IR}
		\vspace{-0.5cm}
	\end{center}
\end{figure}

\subsubsection{LDPC design for iterative detection/decoding}\label{sec:LDPC design}
Figure~\ref{fig:exit2} shows the EXIT chart for \textbf{strategy~2} in case 1. The mutual information (MI) curve of the MUD has been obtained for $P/N\!=\!3$~dB, while the considered codes have rate 1/2. Let us first focus on the MI curve of the LDPC code of the DVB-S2 standard. The EXIT chart analysis reveals that the DVB-S2 codes do not fit the detector, which means that codes designed for systems employing single-user detection in an interference-free scenario are not the best choice for the considered MUD schemes. The EXIT chart of \textbf{strategy~1} has similar features. This observation pushed us into the redesign of the LDPC.
\begin{figure}
	\begin{center}
		\includegraphics[width=1.\columnwidth]{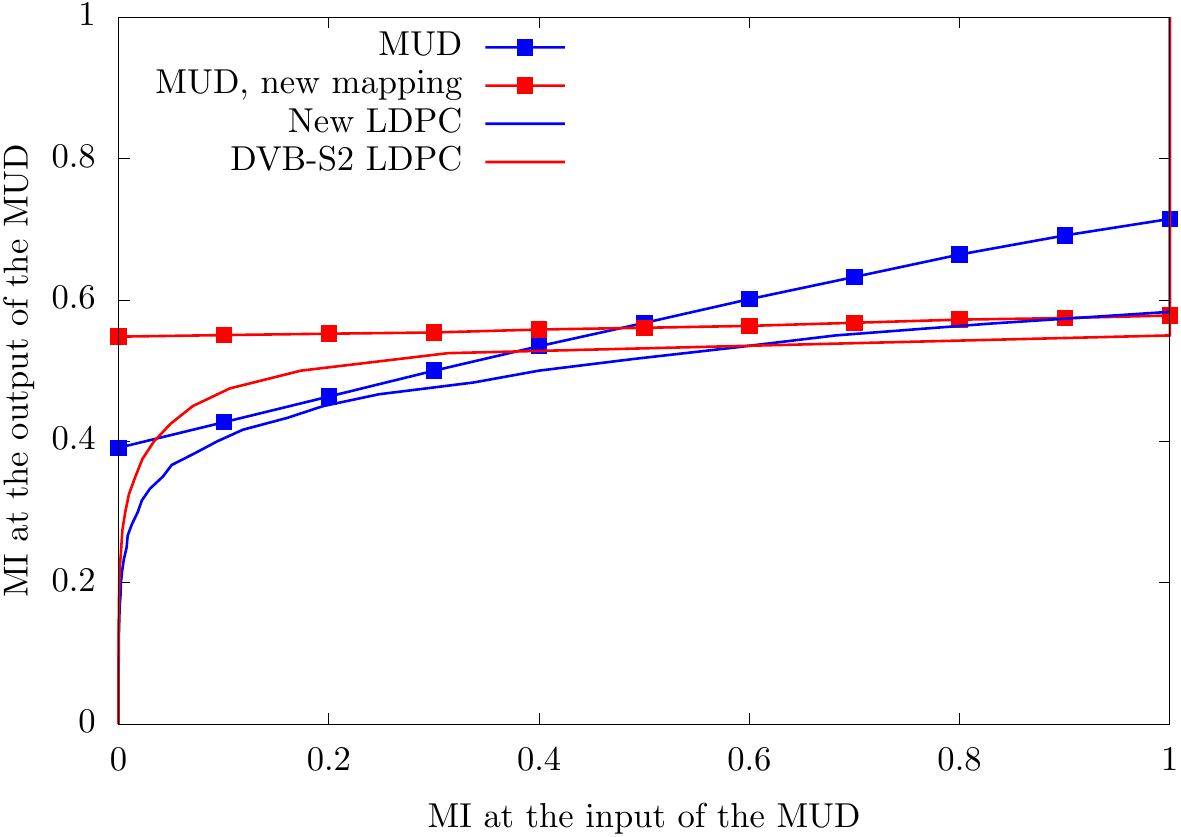}
% 				\vspace{-0.5cm}
		\caption{EXIT chart for \textbf{strategy~2} in case 1 at $P/N=3$~dB.}\label{fig:exit2}
% 		\vspace{-0.5cm}
	\end{center}
\end{figure}

\begin{table}[!t]
 \setlength{\tabcolsep}{4.0pt}
 \scriptsize \centering \caption{Details of designed LDPC codes.}
 \label{t:new codes}

 \vskip -2.0ex %take space out between caption and table
 \begin{tabular}{lccc}
 \hline \\ [-2.0ex]
         &rate & VND distribution  & CND distribution \\
 \hline
 \\ [-2.0ex] \hline  \\ [-2.0ex]
 Strategy 1 & 1/2 &2 (60\%) 3 (31.4\%) 10 (8.6\%) &  6 (100\%) \\\hline
Strategy 2 & 1/2 &2 (60\%) 3 (36.5\%) 20 (3.5\%) &  6 (100\%) \\\hline
Strategy 1 & 3/4 &2 (80\%) 3 (18.3\%) 50 (1.7\%) & 12 (100\%) \\\hline
Strategy 2 & 3/4 &2 (70\%) 3 (28.5\%) 50 (1.5\%) & 12 (100\%) \\\hline
 \end{tabular}
  \vspace{-0.3cm}
 \end{table}

\begin{table}[!t]
 \setlength{\tabcolsep}{4.0pt}
 \scriptsize \centering \caption{BER versus information rate. Performance of new LDPC codes compared to DVB-S2 codes in case~1.}
 \label{t:BERvsIR}
%  \def\Hline{\noalign{\hrule height 2\arrayrulewidth}}
%  \vskip -2.0ex %take space out between caption and table
 \begin{tabular}{lccccc}
 \hline \\ [-2.0ex]
        & rate &IR th. & DVB Code (gap) & New Code (gap)  & phase shift\\
        \hline
 \\ [-2.0ex] \hline  \\ [-2.0ex]
Strategy 1 & 1/2 &6.3 dB   & 8.15 dB (1.85)  & 7.4 dB (1.1)       & $5/16\pi$  \\\hline
Strategy 2 & 1/2 &1.95 dB & 4.1 dB (2.15)    & 3.05 dB (1.1)     & $1/4 \pi$ \\\hline
Strategy 1 & 3/4 & 9.45 dB&  13.3 dB (3.85) & 11.75 dB (1.55) & $1/4 \pi$ \\\hline
Strategy 2 & 3/4 &6.25 dB & 9.75 dB (3.5)   & 7.9 dB (1.65)      & $1/4 \pi$ \\\hline
% Scen. 1 & 2 & 4.4 dB& 6.1 dB (1.7) &  dB () & \\\hline
% Scen. 2 & 2 & 2.55 dB & 4.55 dB (2.0) &  dB ()  & 3.7 dB (1.15) \\\hline
% Scen. 1 & 3 & 2.8 dB& 3.7 dB (0.9) &  dB () & \\\hline
% Scen. 2 & 3 & 3.55 dB & 4.9 dB (1.35) &  dB ()  &  4.55 dB (1.0) \\\hline
 \end{tabular}
  \vspace{-0.5cm}
 \end{table}

The EXIT chart analysis clearly suggests that in our scenario we need an LDPC that is more powerful at the beginning of the iterative process, to have a better curve matching between detector and decoder. This is not surprising, since, in interference-limited channels, a SISO detector is effectively able to mitigate the interference when the information coming from the decoders is somehow reliable. In other words, we mainly need a good head start. We adopt the heuristic technique for the optimization of the degree distribution of the LDPC variable and check nodes proposed in~\cite{teKrAs04}. This method consists of a curve fitting on EXIT charts. We optimize the VND and CND distribution, limiting for simplicity our optimization procedure to codes with uniform check node distribution and only three different variable node degrees. 

Using this approach, for each strategy in case~1 we design a rate-1/2 and a rate-3/4 LDPC code, whose parameters are summarized in Table~\ref{t:new codes}. The EXIT curve of the new LDPC code with rate-1/2 for  \textbf{strategy~2} is shown in Figure~\ref{fig:exit2}. We found other degree distributions with better EXIT curves matching but with poor error floor when used with finite block length.\footnote{In order to improve the finite length performance, the optimization of the code degree distributions could be used jointly with other techniques, e.g. code doping~\cite{LiRyCh08}.}

The codes of length 64800 are then obtained by using the PEG algorithm~\cite{XiBa04}. In the simulations using the new codes with rate 3/4, we decreased the SNR used by the MUD by 0.5 and 0.25 dB  in \textbf{strategy 1} and \textbf{strategy 2}, respectively. In effect, the increase of the noise variance to be set at the receiver improves the performance at high IR where the presence of interference not accounted by the MUD is more critical than for lower IRs. Moreover, for \textbf{strategy~2} we used two different codes, but with the same degree distribution, for the two signals in order to increase the diversity between them.

Table~\ref{t:BERvsIR} summarizes the BER results at IR 1 and 1.5 bit/symbols in terms of convergence threshold, defined as the $P/N$ corresponding to a BER of $10^{-4}$.  We also report the achievable IR limit in $P/N$ obtained through the information-theoretic analysis and the phase shift between the signals $s_1(t)$ and $s_2(t)$. The results show that the gap between the theoretical and the convergence thresholds can be reduced thanks to the new LDPC codes. 

\subsubsection{Joint bit mapping for strategy 2}\label{sec:joint mapping}
After the observation of the poor match between the curves in the EXIT chart, in Section~\ref{sec:LDPC design} we have seen how to improve the threshold by properly changing the code. Here we propose an alternative approach which is focused on the MI curve of the detector. In particular, we propose a joint mapping of the bits of the two signals in~\textbf{strategy~2}, which works exceptionally well in conjunction with the DVB-S2 codes. The idea comes from the fact that transmitting a single signal with Gray mapping gives rise to a practically horizontal EXIT curve for the detector~\cite{te00}, that is exactly what we need if we want to use the codes of the standard.

Let us assume that $M^{(1)}=M^{(2)}=M$ and that the two signals are phase shifted of an angle equal to $\pi/M$. This last choice grants a simple design of the mapping for the resulting constellation and an IR that is close the optimal one. Given two $M$PSK constellations, it is easy to see that, if we rotate one of them by $\pi/M$, the resulting joint constellation is formed by $M/2$ circles, each composed of $2M$ equally spaced points. Two examples are shown in Figure~\ref{fig:multi_const}, for $M$ = 4 (top) and 8 (bottom). 
\begin{figure}
	\begin{subfigure}{1.\columnwidth}
	  \centering
	  \includegraphics[width=1.\columnwidth]{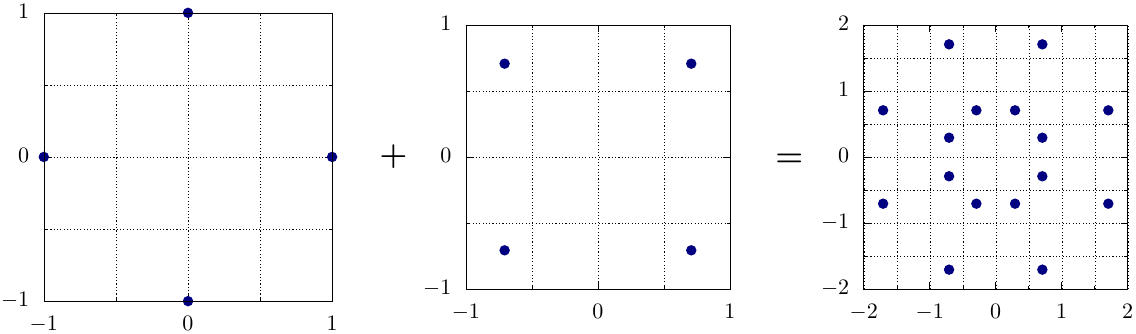}
% 	  \caption{Sum of two QPSK constellations.}
	  \label{fig:2x4psk}
	\end{subfigure}\\
	\begin{subfigure}{1.\columnwidth}
	  \centering
	  \includegraphics[width=1.\columnwidth]{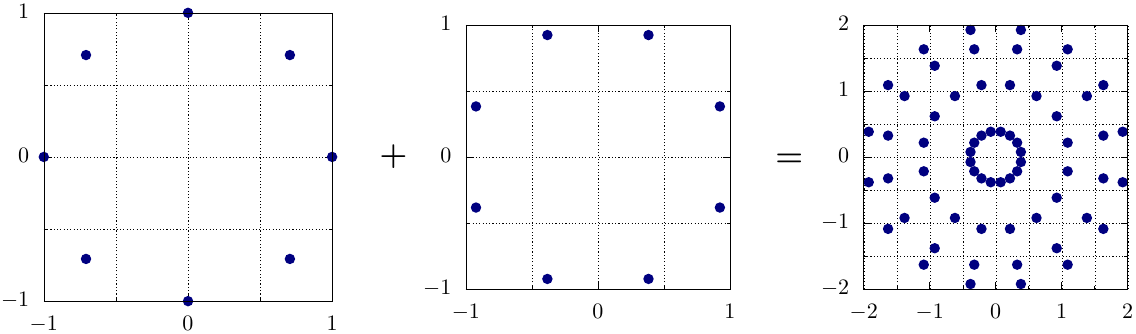}
% 	  \caption{Sum of two 8PSK constellations.}
	  \label{fig:2x8psk}
	\end{subfigure}
	\vspace{-0.5cm}
	\caption{Joint constellations resulting from two QPSK (top) and from two 8PSK (bottom) constellations.}
	\label{fig:multi_const}
\end{figure}
\begin{figure}
	\begin{center}
		\includegraphics[width=1.\columnwidth]{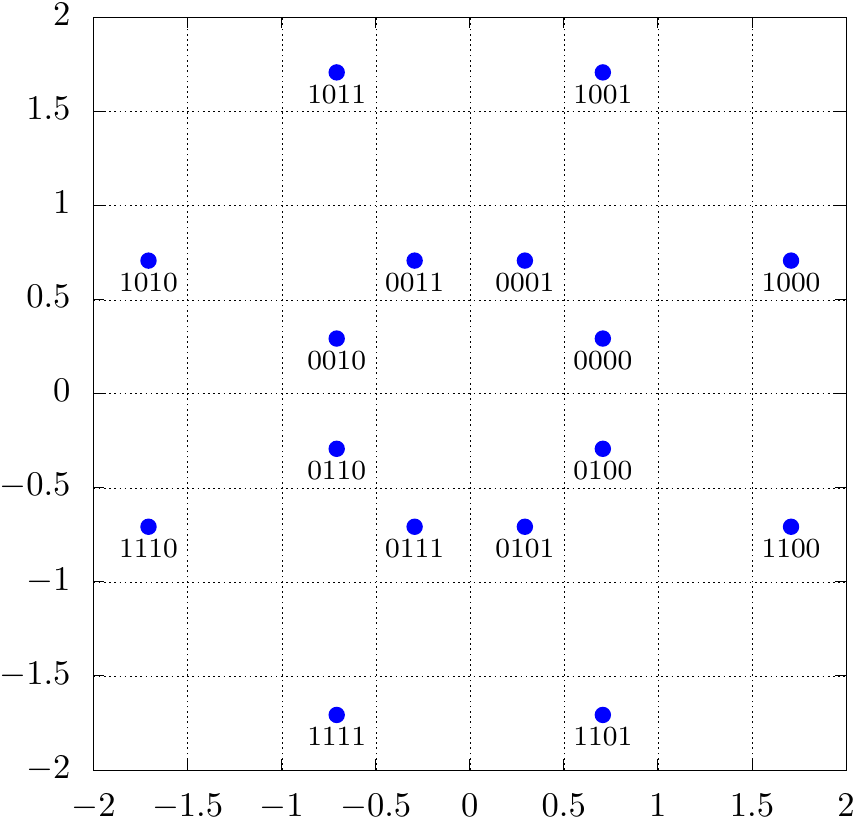}
 		\vspace{-0.1cm}
		\caption{Joint mapping for two QPSK constellations.}\label{fig:mapping}
		\vspace{-0.5cm}
	\end{center}
\end{figure}
We then need to design a good mapping for the joint constellation. Since we have $M^2$ points, we need $\log_2(M^2)$ bits. We choose to use the first $\log_2(M/2)$ bits to identify the circle, and the remaining $\log_2(2M)$ bits to label the points on each circle. Mapping is Gray on each circle and also between adjacent circles. The selected joint mapping is shown in Figure~\ref{fig:mapping} for two QPSK constellations, where the first bit identifies the circle, and the remaining three bits label the points. The EXIT curve of the MUD with joint mapping has smaller slope than that related to the classical mapping, and it is shown in Figure~\ref{fig:exit2}.

A similar approach can be applied in \textbf{strategy~1}, but in this case we can modify only the mapping of signal $s_1(t)$. An example is shown in Figure~\ref{fig:mapping_s1}, for two QPSK constellations: on the left we show the classical mapping, on the right the new mapping, where in red we have the bits of ``User~1'' and in black the bits of ``User~2'', which we are not allowed to modify. We can see that the new mapping is more similar to a Gray mapping than the standard one, in the sense that the distance among adjacent symbols is decreased.

 \begin{figure}[!t]
 	\begin{subfigure}{0.5\columnwidth}
 	  \centering
 	  \includegraphics[width=1.\columnwidth]{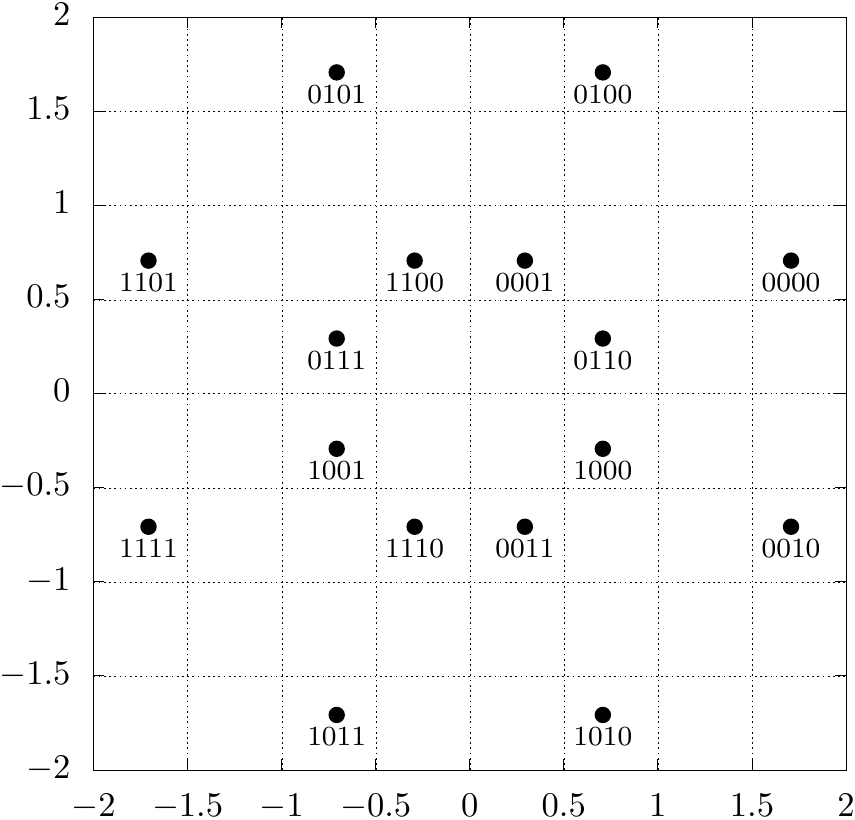}
  	  \caption{Classical mapping.}
 	  \label{fig:2x4psk_classic}
 	\end{subfigure}~
 	\begin{subfigure}{0.5\columnwidth}
 	  \centering
 	  \includegraphics[width=1.\columnwidth]{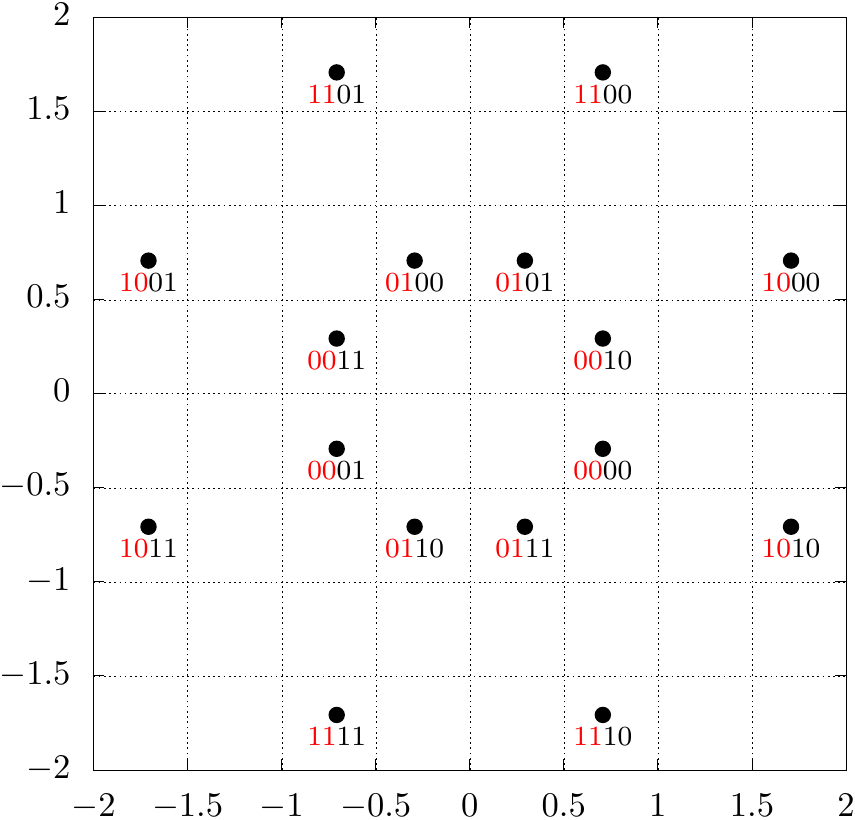}
 	  \caption{New mapping.}
 	  \label{fig:2x4psk_new}
 	\end{subfigure}
 % 	\vspace{-0.5cm}
 	\caption{Joint mappings for \textbf{strategy~1}. Classical (a) and new (b) mapping.}
 	\label{fig:mapping_s1}
 \end{figure}

The BER performance of the joint mapping in \textbf{strategy~2} is reported in Figure~\ref{fig:BER} for the three power profiles in Table~\ref{t:cases}. Both signals $s_1(t)$ and $s_2(t)$ adopt QPSK modulation and DVB-S2 codes with rate $1/2$. The results are compared with the curves of the standard which refer to the classical mapping and with the related  IR thresholds. We observe that the joint mapping improves the performance of the reference curves and the gap with respect to the IR threshold is around 1~dB in all cases.
\begin{figure}
	\begin{center}
		\includegraphics[width=1.\columnwidth]{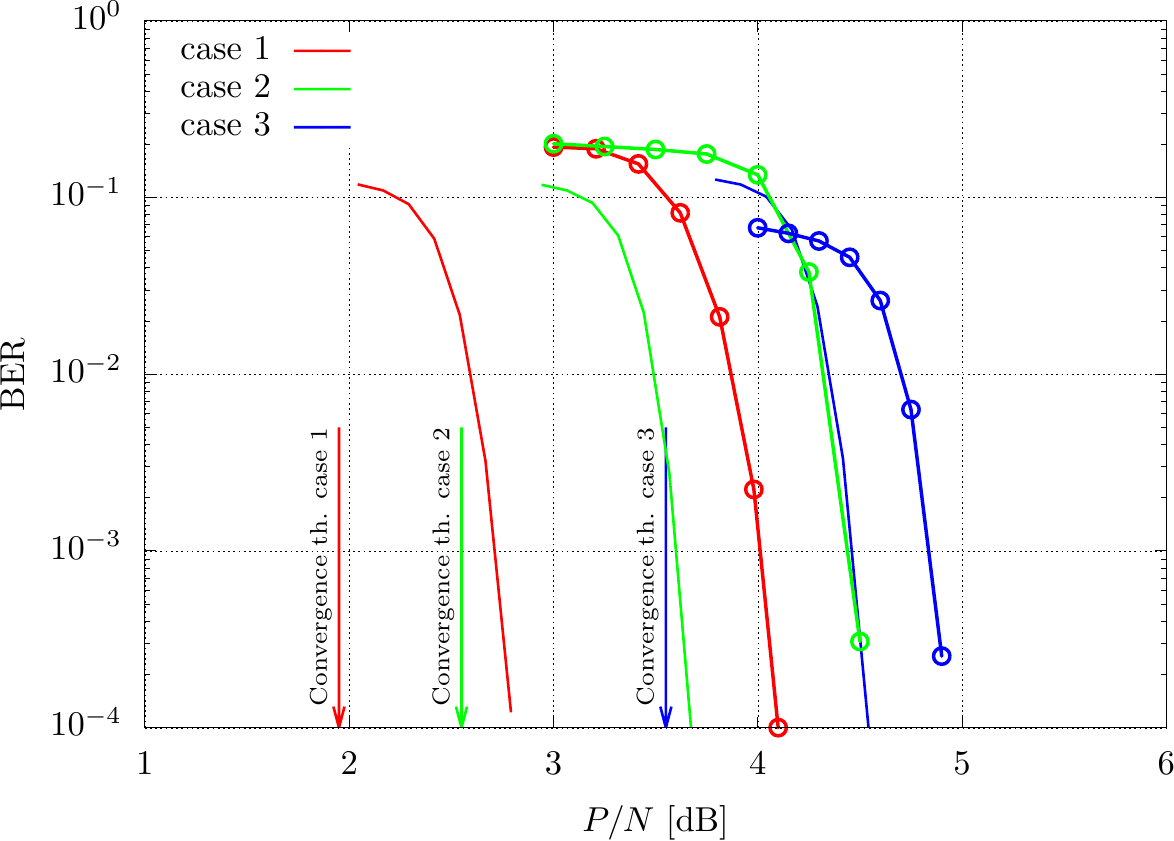}
		\vspace{-0.5cm}
		\caption{BER curves for ``User ~1'' in the case of \textbf{strategy~2} for the three interference patterns. Comparison between joint mapping (continuous line) and classical mapping (continuous line with circles).}\label{fig:BER}
		\vspace{-0.5cm}
	\end{center}
\end{figure}

\section{Conclusions}\label{s:conclusions}

We considered the forward link of a multibeam satellite system, and investigated different transmission/detection strategies  to increase the achievable rate in the presence of strong co-channel interference. As expected, multiuser detection allows a significant gain with respect to single-user detection when the user terminal is close to the edge of the coverage area of its beam. However, it is surprising that the other two strategies requiring modifications at medium access control layer, that based on the use of two transponders to serve consecutively two users
and that based on the use of the Alamouti precoder, can sometimes provide even larger gains, although when the interference is negligible (i.e., when the user terminals is in the center of the beam) a significant loss has to be expected from their use. The conclusive picture is thus complex, since our results show that a transmission/detection strategy which is universally superior to the others does not exist, but the performance depends on several factors, such as the SNR, the interference profile, and the rate of the strongest interferer. This fact outlines the importance of the proposed analysis framework, which can avoid to resort to computationally intensive simulations. Its extension to perform a system analysis, averaging the results on all possible interference profiles within a beam, is also a very interesting subject of investigation.

Finally, we bore evidence that DVB-S2(X) codes, designed for an interference-free scenario, are not suited when a significant interference is present. A proper redesign of the code and/or of the bit mapping can, however, solve the problem.

\bibliographystyle{ieeetr}

\end{document}